\documentclass[pra,superscriptaddress,floatfix, 11pt, letterpaper]{revtex4}
\usepackage{graphicx}
\usepackage{subfigure}
\usepackage{amssymb}
\usepackage{verbatim}
\usepackage{amsmath}
\usepackage{amsfonts}
\usepackage{amscd}
\usepackage{amssymb}
\usepackage{setspace}
\usepackage{hyperref}

\newtheorem{theorem}{Theorem}

\newtheorem{corollary}[theorem]{Corollary}

\newenvironment{proof}{{\noindent{\bf Proof:}}}{$\hfill\Box$}

\newcommand{\ket}[1]{|#1\rangle}
\newcommand{\bra}[1]{\langle#1|}
\newcommand{\tr}{\text{tr}}

\def\Re{\text{Re}}
\def\supp{\text{supp}}

\def\bL{{L}}
\def\cH{\mathcal{H}}

\def\cN{\Phi}

\newcommand{\norm}[1]{ \left\Vert #1 \right\Vert}

\begin{document}

\singlespacing

\title{Sandwiched R\'enyi Divergence Satisfies Data Processing Inequality}

\author{Salman Beigi\\ \emph{\small School of Mathematics,} \emph{\small Institute for Research in Fundamental Sciences (IPM)}\\ \emph{\small P.O. Box 19395-5746, Tehran, Iran}}

\begin{abstract}
Sandwiched (quantum) $\alpha$-R\'enyi divergence has been recently defined in the independent works of Wilde et al. (arXiv:1306.1586) and M\"uller-Lennert et al (arXiv:1306.3142v1). This new quantum divergence has already found applications in quantum information theory. Here we further investigate properties of this new quantum divergence. In particular we show that sandwiched $\alpha$-R\'enyi divergence satisfies the data processing inequality for all values of $\alpha> 1$.  Moreover we prove that $\alpha$-Holevo information, a variant of Holevo information defined in terms of sandwiched $\alpha$-R\'enyi divergence, is super-additive. Our results are based on H\"older's inequality, the Riesz-Thorin theorem and ideas from the theory of complex interpolation. We also employ Sion's minimax theorem.
\end{abstract}

\date{November 15, 2013}

\maketitle

%\tableofcontents

%\narrowtext

%%%%%%%%%%%%%%%%%%%%%%%%%%%%%%%%%%%%%%%%%%%%%%%
\section{A non-commutative R\'enyi divergence}\label{sec:Renyi}

Several entropic quantities that have been shown to be useful in information theory belong to the family of $\alpha$-R\'enyi entropies. For a random variable $X$ with distribution $\{p_i\}$, the $\alpha$-R\'enyi entropy of $X$, for $\alpha> 0$ and $\alpha\neq 1$, is defined by 
\[  H_\alpha(X) = \frac{1}{1-\alpha}\log\left(\sum_i p_i^\alpha \right).\]
In the quantum setting where states are represented by density matrices (positive semi-definite matrices with normalized trace)
$\alpha$-R\'enyi entropy can be defined by 
\[
H_\alpha(\rho) = \frac{1}{1-\alpha}\log\left( \tr\rho^\alpha  \right),
\]
which reduces to the classical R\'enyi entropy when $\rho$ is a diagonal matrix. The limiting cases of R\'enyi entropy when $\alpha\rightarrow 1$ and $\alpha\rightarrow \infty$ are known to be equal to the Shannon entropy and min-entropy respectively.

Likewise a family of R\'enyi divergences can be defined. For two random variables $X, Y$ with distributions $\{p_i\}$ and $\{q_i\}$ respectively, their $\alpha$-R\'enyi divergence for $\alpha>0$ and $\alpha\neq 1$ is defined as follows. If there exists $i$ such that $q_i=0$ but $p_i\neq 0$ then $D_{\alpha}(X||Y)=\infty$. Otherwise,
\begin{align}\label{eq:crd}
D_\alpha(X||Y) = \frac{1}{\alpha -1} \log\left( \sum_i p_i^{\alpha}q_i^{1-\alpha}  \right),
\end{align}
which is equal to 
\begin{align}\label{eq:crdq}
\frac{1}{\alpha-1}\log\left( \tr(\rho^\alpha\sigma^{1-\alpha}) \right),
\end{align}
if we let $\rho$ and $\sigma$ be the diagonal matrices with diagonal entries $\{p_i\}$ and $\{q_i\}$ respectively.
R\'enyi divergence had been generalized to the quantum setting based on the above equation, and appears in the quantum Chernoff bound~\cite{Chernoff} for the range of $0< \alpha<1$. See also~\cite{MosonyiHiai} for other applications of this R\'enyi divergence.  

To capture the non-commutative essence of the quantum theory, another generalization of R\'enyi divergence to the quantum setting was recently proposed in \cite{Wildeetal} and \cite{Lennertetal} (and before that in the talks by Tomamichel \cite{TMarco} and Fehr \cite{Fehr}):
\begin{align}\label{eq:QRD}
D_{\alpha}(\rho||\sigma) = \begin{cases}
\frac{1}{\alpha-1}\log\left(   \tr(\sigma^{\frac{1-\alpha}{2\alpha}}\rho\sigma^{\frac{1-\alpha}{2\alpha}})^{\alpha}    \right)   & \quad\quad \supp(\rho) \subseteq \supp(\sigma)\\
\infty & \quad \quad\quad \quad \text{otherwise},
\end{cases}
\end{align}
where $\rho, \sigma$ are density matrices and $\supp(\rho)$ is the support of $\rho$, i.e., the span of eigenvectors of $\rho$ corresponding to non-zero eigenvalues. This divergence reduces to~\eqref{eq:crdq} when $\rho$ and $\sigma$ commute.
This new divergence is called \emph{sandwiched} R\'enyi relative entropy in \cite{Wildeetal} and quantum R\'enyi divergence in \cite{Lennertetal}. To avoid confusions with the previous R\'enyi divergence we prefer to employ the first name in this paper and call it sandwiched R\'enyi relative entropy or sandwiched R\'enyi divergence.

Sandwiched R\'enyi divergence has already been shown to be useful in quantum information theory. Based on the framework of Sharma and Warsi \cite{SharmaWarsi} (which itself is based on~\cite{PV}), sandwiched R\'enyi divergence is used in \cite{Wildeetal} to prove a strong converse for the classical capacity of entanglement-breaking channels.

To be a useful divergence, one would expect that sandwiched R\'enyi divergence satisfies some properties including the data processing inequality. Some of these properties have been studied in~\cite{Wildeetal, Lennertetal} and are proved to hold especially when $\alpha=1/2$, $1<\alpha\leq 2$ and $\alpha=\infty$. These properties include 
\begin{enumerate}
\item[(a)] \textbf{Positivity and equality condition}: $D_{\alpha}(\rho||\sigma)\geq 0$ and equality holds if and only if $\rho=\sigma$.
\item[(b)] \textbf{Data processing inequality}: for a quantum channel $\Phi$,  $D_{\alpha}(\rho||\sigma)\geq D_{\alpha}( \Phi(\rho)||\Phi(\sigma))$. 
\end{enumerate}

The data processing inequality is proved in \cite{Wildeetal, Lennertetal}  for the range of $1<\alpha \leq 2$ using Lieb's concavity theorem. Moreover the equality condition is shown to hold for the same range of $\alpha$. These two properties for other values of $\alpha$ are conjectured to hold in \cite{Lennertetal}.

In this paper we prove the equality and positivity condition for all positive $\alpha\neq 1$ (see Theorem~\ref{thm:positive}) and the data processing inequality for all values of $\alpha>1$ (see Theorem~\ref{thm:data-processing}).

There are two more conjectures in \cite{Lennertetal} about sandwiched R\'enyi divergence, the first of which is monotonicity in $\alpha$: 
\begin{enumerate}
\item[(c)] \textbf{Monotonicity in $\alpha$}: $\alpha\mapsto D_{\alpha}(\rho||\sigma)$ is increasing.
\end{enumerate}

This conjecture is proved in \cite{Lennertetal} in the special case where $\rho$ is rank-one. 
Here we prove monotonicity in the general case for $\alpha>1$ (see Theorem~\ref{thm:monotonicity}).

From the definitions it is clear that $\alpha$-R\'enyi entropy can be expressed in terms of sandwiched R\'enyi divergence
\[
H_\alpha(\rho) = - D_\alpha(\rho|| I),
\]
where $I$ denotes the identity operator. Moreover, the conditional (von Neumann) entropy can be defined in terms of KL-divergence as follows. For a bipartite state $\rho_{AB}$ we have 
\[
H(A|B)_{\rho} = -\inf_{\sigma_B} D(\rho_{AB}|| I\otimes \sigma_B).
\] 
A similar equality holds for conditional min-entropy in terms of quantum relative max-entropy. Based on these observations quantum conditional R\'enyi entropy is defined in \cite{Lennertetal} by
\begin{align}\label{eq:conditionalRE}
H_{\alpha}(A|B)_{\rho} := -\inf_{\sigma_B} D_{\alpha}(\rho_{AB} || I\otimes \sigma_B). 
\end{align}
Then the following is conjectured.
\begin{enumerate}
\item[(d)] \textbf{Duality}: For all $1/2\leq  \alpha, \beta \leq \infty$, $\alpha,\beta\neq 1$, such that $\frac{1}{\alpha} + \frac{1}{\beta}=2$, and all tripartite \emph{pure} states $\rho_{ABC}$ we have $H_{\alpha}(A|B) = -H_{\beta}(A|C)$.
\end{enumerate}

The special case of this conjecture when $\rho_{AB}$ is pure, is proved in \cite{Lennertetal}. Here we give a proof for the general case (see Theorem~\ref{thm:duality}).

We also answer an open question raised in~\cite{Wildeetal} about the super-additivity of a quantity called $\alpha$-Holevo information (see Theorem~\ref{thm:chi-super-additive}).

To prove these results we mostly employ properties of Schatten norms. In particular we use H\"older's inequality and its generalizations for proving (a). For (b) and (c), we use the Riesz-Thorin theorem and ideas from the theory of complex interpolation. Finally (d) is proved based on H\"older's duality and Sion's minimax theorem. 

In the following two sections we review H\"older's inequalities and prove the Riesz-Thorin theorem. The main results are stated and proved in Sections~\ref{sec:proofs} and~\ref{sec:chi}.

\emph{Note:} After completion of this work we discovered that Frank and Lieb \cite{FrankLieb} have also proved the data processing inequality (property (b)). Their proof (that is different from ours) works for all values of $\alpha\geq 1/2$ unlike ours which works only for $\alpha>1$. Moreover, monotonicity (property (c)) and duality (property (d)) have been proved in the updated paper of M\"uller-Lennert et al~\cite{updated}. Their proof of duality is similar to ours, but they have a different proof for monotonicity.

%**********************

\section{H\"older's inequalities}\label{sec:holder}

For a finite dimensional Hilbert space $\cH$, the set of linear operators is denoted by $\bL(\cH)$. The adjoint of $X\in \bL(\cH)$ is denoted by $X^{\dagger}$. The Hilbert-Schmidt inner product on $\bL(\cH)$ is defined by 
$$\langle X, Y\rangle := \tr(X^{\dagger}Y),$$
where $\tr(\cdot)$ is the usual trace function. Throughout this paper for a hermitian (self-adjoint) operator $X$, by $X^{-1}$ we mean the inverse restricted to $\supp(X)$, so $XX^{-1}= X^{-1}X$ equals to the orthogonal projection on $\supp(X)$.

For $X\in \bL(\cH)$ and real $p\neq 0$ define
$$\|X\|_p = \left(\tr  |X|^p  \right)^{\frac{1}{p}},$$
where $|X| = (X^{\dagger}X)^{1/2}$. Note that by the above convention $\|X\|_p$ is defined even for a negative $p$. We also define 
$$\|X\|_{\infty} = \lim_{p\rightarrow \infty} \|X\|_p,$$
which is the usual operator norm of $X$. From the definition we clearly have $\|UXV\|_p = \|X\|_p$ for unitary operators $U,V$. Moreover, $\|X^{\dagger}\|_p = \|X\|_p$.

It is well-known that $\|\cdot \|_p$ for $1\leq p\leq \infty$ satisfies the triangle inequality and is a norm. $\bL(\cH)$ equipped with this norm is denoted by $\bL_p(\cH)$. 

For $p\neq 0$, we let $p'$ be the H\"older conjugate of $p$, i.e., $p'$ is defined by  
\begin{align}\label{eq:conjugate}
\frac{1}{p}+ \frac{1}{p'}=1.
\end{align}
H\"older's inequality states that 
$$\|XY\|_1\leq \|X\|_p \|Y\|_{p'}, \quad\quad\quad\quad \quad 1\leq p\leq \infty,$$
which also implies $|\tr(XY)| \leq \|X\|_p\|Y\|_{p'}$. Using this inequality it is easy to see that for $1\leq p\leq \infty$, the dual space of $\bL_{p}(\cH)$ is $\bL_{p'}(\cH)$. In other words we have
\begin{align}\label{eq:holder-dual}
\|X\|_p= \sup_{Y: \|Y\|_{p'}=1} | \langle Y, X\rangle|, \quad\quad\quad\quad \quad   1\leq p\leq \infty.
\end{align}

H\"older's inequality belongs to a richer family of inequalities. For every $p, q, r>0$ with 
$\frac{1}{r} = \frac{1}{p} + \frac{1}{q}$ we have (see for example Exercise IV.2.7 of \cite{Bhatia-M})
\begin{align} \label{eq:holder}
\|XY\|_r\leq \|X\|_p \|Y\|_q.
\end{align}
Moreover, equality holds in \eqref{eq:holder} if and only $|X|^p$ and $|Y^{\dagger}|^q$ are proportional. 
Then by a simple induction, for every $p_1, \dots, p_k, r>0$ with $\frac 1 r= \frac 1{p_1} + \cdots + \frac{1}{p_k}$ we obtain 
\begin{align}\label{eq:m-holder}
\|X_1\cdots X_k\|_r \leq \|X_1\|_{p_1}\cdots \|X_k\|_{p_k}.
\end{align}

From this inequality and the fact that $\|X^{-1}\|_{-p} = \|X\|_{p}^{-1}$ the following \emph{reverse} H\"older inequality is derived. Let $r>0$ and $p_1, \dots, p_k$ be such that $\frac 1 r= \frac 1{p_1} + \cdots + \frac{1}{p_k}$ and that \emph{exactly} one of $p_i$'s is positive and the rests are negative. Then
\begin{align}\label{eq:m-reverseholder}
\|X_1\cdots X_k\|_r \geq \|X_1\|_{p_1}\cdots \|X_k\|_{p_k}.
\end{align}
In particular we have
\begin{align}\label{eq:reverseholder}
\|X\|_p \|Y\|_{p'}\leq \|XY\|_1, \quad\quad\quad\quad \quad   0< p<1.
\end{align}
Moreover if $X$ is positive semi-definite we have
\begin{align}\label{eq:inverse-holder-duality}
\|X\|_p = \inf_{Y\geq 0, \|Y\|_{p'}=1}   \tr(XY), \quad\quad\quad\quad \quad   0< p<1.
\end{align}

We finish this section by introducing one more notation. Let $\sigma\in \bL(\cH)$ be positive semi-definite, and define the super-operator
$\Gamma_\sigma(X):= \sigma^{1/2}X\sigma^{1/2}$. Define
$$\|X\|_{p,\sigma}:= \| \Gamma_\sigma^{\frac 1 p}(X) \|_p = \| \sigma^{\frac 1{2p}}X\sigma^{\frac 1 {2p}}  \|_p.$$
When $\sigma$ is positive definite (and then full-rank) a simple manipulation verifies that $\|\cdot\|_{p, \sigma}$ is a norm for $1\leq p\leq \infty$, and also the following duality holds.
$$\|X\|_{p, \sigma}= \sup_{Y: \|Y\|_{p', \sigma}=1}  |\langle Y, X\rangle_{\sigma}|, \quad\quad\quad\quad \quad   1\leq p\leq \infty,$$
where $\langle Y, X\rangle_{\sigma} = \tr((Y^{\dagger}\Gamma_\sigma(X)) = \tr(Y^{\dagger}\sigma^{1/2}X\sigma^{1/2})$. 
The space $\bL(\cH)$ equipped with this norm is denoted by $\bL_{p, \sigma}(\cH)$.

%**************************************
\section{Riesz-Thorin theorem}\label{sec:Riesz-Thorin}

Most of the proofs in this paper are based on the theory of complex interpolation, especially the Riesz-Thorin theorem for which we refer to the textbook \cite{BerghL} and lecture notes \cite{Xu, Lunardi}. This theory has already found applications in quantum information theory \cite{PWPR, BD}. 
Here to obtain self-contained proofs, instead of directly referring to this theory we prefer to give a proof of the Riesz-Thorin theorem in the special case that is more relevant to quantum information theory. This proof is based on Hadamard's three-line theorem (see \cite[page 33]{ReedSimon}).

Define 
$$S=\{z\in \mathbb C: 0\leq \Re\, z\leq 1\},$$
where  $\Re\, z \in \mathbb R$ denotes the real part of the complex number $z\in \mathbb C$.

\begin{theorem}\label{thm:hadamard}\emph{(Hadamard's three-line theorem \cite{ReedSimon})} Let $f: S\rightarrow \mathbb C$ be a bounded function that is holomorphic in the interior of $S$ and continuous on the boundary. For $k=0, 1$ let
$$M_k=\sup_{t\in \mathbb R} |f(k+it)|.$$
Then for every $0\leq \theta\leq 1$ we have $|f(\theta)|\leq M_0^{1-\theta}M_1^{\theta}$.
\end{theorem}

A map $F:S\rightarrow \bL(\cH)$ is call holomorphic (continuous, bounded) if the corresponding functions to matrix entries is holomorphic (continuous, bounded). 
The following theorem is a generalization of Hadamard's three-line theorem.

 \begin{theorem}\label{thm:banach-hadamard}
 Let $F:S\rightarrow \bL(\cH)$ be a bounded map that is holomorphic in the interior of $S$ and continuous on the boundary. Let $\sigma\in \bL(\cH)$ be positive definite. Assume that $1\leq p_0\leq  p_1\leq \infty$ and for $0< \theta<1 $ define $p_0\leq p_\theta\leq p_1$ by 
 \begin{align}\label{eq:p-theta}
 \frac{1}{p_\theta} = \frac{1-\theta}{p_0} + \frac{\theta}{p_1}.
 \end{align}
 For $k=0, 1$ define 
 $$M_k = \sup_{t\in \mathbb R} \|F(k+it)\|_{p_k, \sigma}.$$
 Then we have 
 $$\|F(\theta)\|_{p_{\theta}, \sigma}\leq M_0^{1-\theta}M_1^{\theta}.$$
 \end{theorem}

\begin{proof} Let $X$ be such that $\|X\|_{p'_{\theta}, \sigma} = 1$ and $\|F(\theta)\|_{p_{\theta}, \sigma} = \langle X^{\dagger}, f(\theta)\rangle_{\sigma}.$ Using $\|X\|_{p'_{\theta}, \sigma}=\|\Gamma_\sigma^{1/p'_{\theta}}(X)\|_{p'_{\theta}}=1$, the singular value decomposition of $\Gamma_\sigma^{1/p'_{\theta}}(X)$ has the form 
$$\Gamma_\sigma^{1/p'_{\theta}}(X) = UD^{\frac{1}{p'_{\theta}}}V,$$
where $U, V$ are unitary and $D$ is diagonal with non-negative entries and $\tr(D)=1$.
Define 
$$X(z) = \Gamma_{\sigma}^{-(\frac{1-z}{p'_0} + \frac{z}{p'_1})} \left( U D^{(\frac{1-z}{p'_0} + \frac{z}{p'_1})}  V    \right)= {\sigma}^{-(\frac{1-z}{2p'_0} + \frac{z}{2p'_1})}     \left( U D^{(\frac{1-z}{p'_0} + \frac{z}{p'_1})}  V    \right)   {\sigma}^{-(\frac{1-z}{2p'_0} + \frac{z}{2p'_1})} .$$
Observe that the map $z\mapsto X(z)$ is holomorphic, and $X(\theta)=X$.

Now define
$$g(z) = \langle X(it)^{\dagger}, F(z)\rangle_{\sigma} = \tr\left( X(z) \sigma^{1/2} F(z)\sigma^{1/2}  \right).$$
$g:S\rightarrow \mathbb C$ satisfies assumptions of Hadamard's three-line theorem. Thus we have
\begin{align*}
\|F(\theta)\|_{p_{\theta}, \sigma} & = \langle X^{\dagger}, F(\theta)\rangle_{\sigma}\\
& = |g(\theta)|\\
&\leq \left(  \sup_{t\in \mathbb R}    |g(it)|    \right)^{1-\theta} \left(   \sup_{t\in \mathbb R}  |g(1+it)|   \right)^{\theta}\\
& = \left( \sup_{t\in \mathbb R}   |\langle X(it)^{\dagger}, F(it)\rangle_{\sigma}|        \right)^{1-\theta} \left( \sup_{t\in \mathbb R}   |\langle X(1+it)^{\dagger}, F(1+it)\rangle_{\sigma} |       \right)^{\theta}\\
&\leq \left(  \sup_{t\in \mathbb R}   \| X(it)  \|_{p'_0, \sigma} \|  F(it) \|_{p_0, \sigma}    \right)^{1-\theta} \left(  \sup_{t\in \mathbb R}   \| X(1+it)  \|_{p'_1, \sigma} \|  F(1+it) \|_{p_1, \sigma}    \right)^{\theta},
\end{align*}
where in the last line we use H\"older's inequality. 
By definition we have $\|X(it)\|_{p'_0, \sigma} = \|\Gamma_\sigma^{p'_0}(X(it))\|_{p'_0}$ and 
$$\Gamma_\sigma^{p'_0}(X(it)) = \sigma^{\frac{it}{2p'_0}-\frac{it}{2p'_1}} \left( U D^{(\frac{1-it}{p'_0} + \frac{it}{p'_1})}  V    \right) \sigma^{\frac{it}{2p'_0}-\frac{it}{2p'_1}}= U_t D^{\frac{1}{p'_0}}V_t,$$
where $U_t =  \sigma^{\frac{it}{2p'_0}-\frac{it}{2p'_1}}  U D^{-\frac{it}{p'_0} + \frac{it}{p'_1})}$ and $V_t = V \sigma^{\frac{it}{2p'_0}-\frac{it}{2p'_1}},$ are unitary. As a result, $\|X(it)\|_{p'_0, \sigma}=1$ for every $t\in \mathbb R$. We similarly have $\|X(1+it)\|_{p'_1,\sigma}=1$. Therefore, 
$$\|F(\theta)\|_{p_\theta, \sigma} \leq \left(  \sup_{t\in \mathbb R}    \|F(it) \|_{p_0, \sigma}    \right)^{1-\theta} \left(  \sup_{t\in \mathbb R}    \|  F(1+it) \|_{p_1, \sigma}    \right)^{1-\theta} = M_0^{1-\theta}M_1^{\theta}.$$

\end{proof}

Using this theorem one can indeed show that $\bL_{p_{\theta}, \sigma}(\cH)$ is the complex interpolation space between $\bL_{p_0, \sigma}(\cH)$ and $\bL_{p_1, \sigma}(\cH)$. See \cite{PisierXu} and references there for more details. See also \cite{RicardXu, Conde} for similar results.

\begin{corollary}\label{corol:theta}
Let $1\leq p_0< p_1\leq \infty$ and $0<\theta<1$, and define $p_{\theta}$ by \eqref{eq:p-theta}. Then for every positive definite $\sigma\in \bL(\cH)$ and $X\in \bL(\cH)$ we have
$$\|X\|_{p_{\theta}, \sigma}\leq \|X\|_{p_0, \sigma}^{1-\theta}\|X\|_{p_1, \sigma}^{\theta}.$$
\end{corollary}

\begin{proof}
In Theorem~\ref{thm:banach-hadamard} take the constant map $F(z)=X$.
\end{proof}

\vspace{.25in}

We need one more notation to state the Riesz-Thorin theorem. Let $\Phi:\bL(\cH)\rightarrow \bL(\cH')$ be a linear super-operator. Then for each $1\leq p, q\leq \infty$ and positive definite $\sigma\in \bL(\cH)$ and $\sigma'\in \bL(\cH')$ we may consider $\Phi$ as an operator from the normed space $\bL_{p, \sigma}(\cH)$ to $\bL_{q, \sigma'}(\cH')$. Then the super-operator norm of $\Phi$ is defined by
\begin{align}\label{eq:sup-norm-operator}
\|\Phi\|_{(p,\sigma)\rightarrow (q, \sigma')}  = \sup_{X\neq 0}  \frac{\|\Phi(X)\|_{q, \sigma'}}{\|X\|_{p, \sigma}}.
\end{align}
From the definition it is clear that for every $X$ we have
$$\|\Phi(X)\|_{q, \sigma'} \leq \|\Phi\|_{(p,\sigma)\rightarrow (q, \sigma')} \|X\|_{p, \sigma}.$$

An important property of the norm~\eqref{eq:sup-norm-operator} is that if $\Phi$ is completely-positive, then the supremum is attained at a positive semi-definite $X$~\cite{Watrous, Audenaert}. Here we use the fact that the composition of a completely-positive map with $\Gamma_{\sigma}$ is also completely-positive.

\begin{theorem}\label{thm:riesz-thorin} \emph{(Riesz-Thorin theorem for $\bL_{p, \sigma}$ spaces)} Let $\Phi:\bL(\cH)\rightarrow \bL(\cH')$ be a linear super-operator. Assume that $1\leq p_0 \leq p_1\leq \infty$ and $1\leq q_0\leq q_1\leq \infty$.  Let $0\leq \theta\leq 1$ and define $p_\theta$ and similarly $q_\theta$ using \eqref{eq:p-theta}. Finally assume that $\sigma\in \bL(\cH)$ and $\sigma'\in \bL(\cH')$ are positive definite. Then we have
$$\|\Phi\|_{(p_\theta, \sigma)\rightarrow(q_\theta, \sigma')} \leq \|\Phi\|_{(p_0, \sigma)\rightarrow(q_0, \sigma')}^{1-\theta} \|\Phi\|_{(p_1, \sigma)\rightarrow(q_1, \sigma')}^{\theta}.$$
\end{theorem}

\begin{proof} It suffices to show that for every $X\in \bL(\cH)$ with $\|X\|_{p_\theta, \sigma}=1$ we have 
$$\|\Phi(X)\|_{q_\theta, \sigma'}\leq \|\Phi\|_{(p_0, \sigma)\rightarrow(q_0, \sigma')}^{1-\theta} \|\Phi\|_{(p_1, \sigma)\rightarrow(q_1, \sigma')}^{\theta}.$$
As in the proof of Theorem~\ref{thm:banach-hadamard} such an $X$ has the form $X= \Gamma_{\sigma}^{-\frac{1}{p_{\theta}}}\left( UD^{\frac 1 {p_{\theta}}} V \right)$ where $U, V$ are unitary and $D$ is diagonal with non-negative entries and $\tr(D)=1$. Now define 
$$X(z) = \Gamma_{\sigma}^{-(\frac{1-z}{p_0} + \frac{z}{p_1})} \left(U D^{(\frac{1-z}{p_0} + \frac{z}{p_1}))}  V \right),$$
and let $F:S\rightarrow \mathbb C$, $F(z) = \Phi(X(z))$. Then by Theorem~\ref{thm:banach-hadamard} we have
\begin{align}\label{eq:RT-bound}
\|\Phi(X)\|_{q_{\theta}, \sigma'}=\|\Phi(X(\theta))\|_{q_{\theta}, \sigma'} \leq \left( \sup_{t\in \mathbb R}\| \Phi(X(it))\|_{q_0, \sigma'}\right)^{1-\theta}
\left( \sup_{t\in \mathbb R}\| \Phi(X(1+it))\|_{q_1, \sigma'}\right)^{\theta}.
\end{align} 
Observe that, by the definition of the operator norm,  we have
$$ \| \Phi(X(it))\|_{q_0, \sigma'} \leq \|\Phi\|_{(p_0, \sigma)\rightarrow (q_0, \sigma')} \|X(it)\|_{p_0, \sigma}.$$
On the other hand, similar to the argument presented in the proof of Theorem~\ref{thm:banach-hadamard}, $ \|X(it)\|_{p_0, \sigma}=1$. As a result 
$$\sup_{t\in \mathbb R}\| \Phi(X(it))\|_{q_0, \sigma'} \leq \|\Phi\|_{(p_0, \sigma)\rightarrow (q_0, \sigma')},$$
and similarly 
$$\sup_{t\in \mathbb R}\| \Phi(X(1+it))\|_{q_1, \sigma'} \leq \|\Phi\|_{(p_1, \sigma)\rightarrow (q_1, \sigma')}.$$
The proof finishes by using these two bounds in \eqref{eq:RT-bound}.

\end{proof}

%*****************
\section{Statements and proofs of the main results}\label{sec:proofs}

Using notations developed in Section~\ref{sec:holder} sandwiched R\'enyi divergence~\eqref{eq:QRD} can equivalently be defined by
\begin{align}\label{eq:QRD-norm-1}
D_{\alpha}(\rho||\sigma) = \begin{cases}
\alpha'\log\|\sigma^{-\frac{1}{2\alpha'}}\rho\sigma^{-\frac{1}{2\alpha'}} \|_{\alpha}   & \quad\quad \supp(\rho) \subseteq \supp(\sigma)\\
\infty & \quad \quad\quad \quad \text{otherwise}.
\end{cases}
\end{align}
Here we use the fact that $\alpha'$ the H\"older conjugate of $\alpha$ define by~\eqref{eq:conjugate} is equal to 
$$\alpha' = \frac{\alpha}{\alpha -1}.$$

In the following we also use
\begin{align}\label{eq:norm-eq}
\|   \sigma^{-\frac{1}{2\alpha'}}\rho\sigma^{-\frac{1}{2\alpha'}}  \|_{\alpha} = \|\Gamma_{\sigma}^{-\frac{1}{\alpha'}}(\rho)\|_{\alpha} = \|\Gamma_{\sigma}^{-1} (\rho)\|_{\alpha, \sigma}
\end{align}

We now have all the required tools to prove properties of sandwiched R\'enyi divergence stated in Section~\ref{sec:Renyi}.

\begin{theorem}\label{thm:positive} \emph{(Positivity and equality condition)} $D_\alpha(\rho||\sigma)\geq 0$ for density matrices $\rho, \sigma$ and all positive $\alpha\neq 1$. Moreover, equality holds if and only if $\rho = \sigma$.
\end{theorem}

\begin{proof}
Using expression~\eqref{eq:QRD-norm-1} for sandwiched R\'enyi divergence we need to show that $\|\sigma^{-\frac{1}{2\alpha'}}\rho \sigma^{-\frac{1}{2\alpha'}}\|_{\alpha} \geq 1$ when $\alpha>1$, and $\|\sigma^{-\frac{1}{2\alpha'}}\rho \sigma^{-\frac{1}{2\alpha'}}\|_{\alpha} \leq 1$ when $\alpha<1$.

Observe that $\|\sigma^{\frac{1}{2\alpha'}}\|_{2\alpha'} = [\tr (\sigma)]^{1/(2\alpha')}=1$. Moreover, $\frac{1}{2\alpha'} + \frac{1}{\alpha} + \frac{1}{2\alpha'} =1$. Thus for $\alpha>1$ by~\eqref{eq:m-holder} we have
\begin{align*}
\|\sigma^{-\frac{1}{2\alpha'}}\rho \sigma^{-\frac{1}{2\alpha'}}\|_{\alpha} & =\|\sigma^{\frac{1}{2\alpha'}}\|_{2\alpha'}  \|\sigma^{-\frac{1}{2\alpha'}}\rho \sigma^{-\frac{1}{2\alpha'}}\|_{\alpha}\|\sigma^{\frac{1}{2\alpha'}}\|_{2\alpha'} \\
& \geq   \|\sigma^{\frac{1}{2\alpha'}} \left(\sigma^{-\frac{1}{2\alpha'}}\rho \sigma^{-\frac{1}{2\alpha'}}\right) \sigma^{\frac{1}{2\alpha'}}\|_{1} \\ & = \|\rho\|_1\\
& =1.
\end{align*}
The case $\alpha<1$ is similar and is proved using~\eqref{eq:m-reverseholder}. 

The equality condition is simply a consequence of the equality condition in H\"older's inequality~\eqref{eq:holder}. In fact equality implies that $\sigma$ and $(\sigma^{-\frac{1}{2\alpha'}}\rho \sigma^{-\frac{1}{2\alpha'}})^{\alpha}$ are proportional. This in particular implies that the density matrices $\rho$ and $\sigma$ commute, which implies that $\sigma$ and $\rho$ are proportional and then equal.
\end{proof}

\begin{theorem}\label{thm:data-processing} \emph{(Data processing inequality)} For all density matrices $\rho, \sigma$, and $\alpha>1$, and CPTP map (quantum channel) $\Phi$ we have 
\begin{align}\label{eq:data-processing}
D_{\alpha}(\rho||\sigma)\geq D_{\alpha}(\Phi(\rho)||\Phi(\sigma)).
\end{align}
\end{theorem}

\begin{proof} There is nothing to prove when $D_{\alpha}(\rho||\sigma)=\infty$. So let us assume that $\supp(\rho)\subseteq \supp(\sigma)$. 

Since $\alpha>1$ and the logarithm function is increasing, \eqref{eq:data-processing} is equivalent to 
$$\| \Gamma_{\sigma}^{-1}(\rho) \|_{\alpha, \sigma} \geq \| \Gamma_{\Phi(\sigma)}^{-1}(\Phi(\rho))  \|_{\alpha, \Phi(\sigma)}.$$
Observe that 
$$\Gamma_{\Phi(\sigma)}^{-1}(\Phi(\rho)) = \Gamma_{\Phi(\sigma)}^{-1}\circ \Phi \circ \Gamma_{\sigma} \left( \Gamma_\sigma^{-1}(\rho)  \right).$$
As a result,
$$ \| \Gamma_{\Phi(\sigma)}^{-1}(\Phi(\rho))  \|_{\alpha, \Phi(\sigma)} \leq \|\Gamma_{\Phi(\sigma)}^{-1}\circ \Phi \circ \Gamma_{\sigma}\|_{(\alpha, \sigma)\rightarrow (\alpha, \Phi(\sigma))} \|\Gamma_{\sigma}^{-1}(\rho)\|_{\alpha, \sigma}.$$
Therefore, it is sufficient to prove that 
\begin{align}\label{eq:norm-ineq}
\|\Gamma_{\Phi(\sigma)}^{-1}\circ \Phi \circ \Gamma_{\sigma}\|_{(\alpha, \sigma)\rightarrow (\alpha, \Phi(\sigma))}\leq 1.
\end{align}
Employing Riesz-Thorin theorem (Theorem~\ref{thm:riesz-thorin})we only need to prove this for $\alpha=1$ and $\alpha=\infty$. 

For $\alpha=1$ we have $\|\Gamma_{\Phi(\sigma)}^{-1}\circ \Phi \circ \Gamma_{\sigma}\|_{(1, \sigma)\rightarrow (1, \Phi(\sigma))} = \|\Phi\|_{1\rightarrow 1}=1$ because $\Phi$ is completely-positive and trace preserving. 

For $\alpha=\infty$ we have $\|\Gamma_{\Phi(\sigma)}^{-1}\circ \Phi \circ \Gamma_{\sigma}\|_{(\infty, \sigma)\rightarrow (\infty, \Phi(\sigma))} = \|\Gamma_{\Phi(\sigma)}^{-1}\circ \Phi \circ \Gamma_{\sigma}\|_{\infty\rightarrow \infty}$. On the other hand $\Gamma_{\Phi(\sigma)}^{-1}\circ \Phi \circ \Gamma_{\sigma}$ is a positive map, then by the Russo-Dye theorem (Corollary 2.3.8 of \cite{Bhatia-P})
we have
$$\|\Gamma_{\Phi(\sigma)}^{-1}\circ \Phi \circ \Gamma_{\sigma}\|_{\infty\rightarrow \infty} = \|\Gamma_{\Phi(\sigma)}^{-1}\circ \Phi \circ \Gamma_{\sigma} (I)\|_{\infty} = \|\Gamma_{\Phi(\sigma)}^{-1}\circ \Phi ({\sigma})\|_{\infty} = \|\Phi(\sigma)^{-\frac{1}{2}} \Phi(\sigma)\Phi(\sigma)^{-\frac1 2}\|=1.$$
We are done.

\end{proof}

Here we should mentioned that inequality~\eqref{eq:norm-ineq} has also been proven in \cite{BD} and has other consequences in quantum information theory.

\begin{theorem}\label{thm:monotonicity}\emph{(Monotonicity in $\alpha$)} For all density matrices $\rho, \sigma$, 
the function $\alpha\mapsto D_{\alpha}(\rho||\sigma)$ is increasing for $\alpha>1$.
\end{theorem}

\begin{proof}
Again using the monotonicity of the logarithm function it suffices to prove that for $1<\alpha<\beta$ we have
$$\|\Gamma_{\sigma}^{-1}(\rho)\|_{\alpha, \sigma}^{\alpha'} \leq \|\Gamma_{\sigma}^{-1}(\rho)\|_{\beta, \sigma}^{\beta'}.$$
Note that this inequality is stronger that the statement of the theorem and gives the monotonicity of $\alpha \mapsto \exp(D_{\alpha}(\rho||\sigma))$. 

Since $1<\alpha<\beta$ there exists $0<\theta <1$ such that 
\begin{align}\label{eq:theta-ab}
\frac{1}{\alpha} = (1-\theta) + \frac{\theta}{\beta}.
\end{align} 
Then by Corollary~\ref{corol:theta} we have
$$\|\Gamma_{\sigma}^{-1}(\rho)\|_{\alpha, \sigma} \leq \|\Gamma_{\sigma}^{-1}(\rho) \|_{1, \sigma}^{1-\theta} \|\Gamma_{\sigma}^{-1}(\rho)\|_{\beta, \sigma}^\theta.$$
On the other hand $\|\Gamma_{\sigma}^{-1}(\rho) \|_{1, \sigma} = \|\rho\|_1=1$. Therefore, by raising both sides to the power of $\alpha'$ we arrive at 
$$\|\Gamma_{\sigma}^{-1}(\rho)\|_{\alpha, \sigma}^{\alpha'} \leq  \|\Gamma_{\sigma}^{-1}(\rho)\|_{\beta, \sigma}^{\theta\alpha'}.$$
The proof is finished by observing that \eqref{eq:theta-ab} implies that $\theta\alpha'=\beta'$.

\end{proof}

In the above proof we use Corollary~\ref{corol:theta} for $p_0=1$, $p_1=\beta$ and $p_\theta=\alpha$ to show monotonicity in $\alpha$. Assuming that $p_0>1$ is arbitrary, and following the same proof we obtain some convexity property of $\alpha$-R\'enyi divergence.

\begin{theorem} Let $1\leq \alpha< \beta<\gamma$ and define $\theta$ by
$$\frac{1}{\beta} = \frac{1-\theta}{\alpha} + \frac{\theta}{\gamma}.$$
Then for every $\rho, \sigma$ we have
$$\frac{1}{\beta'}D_{\beta}(\rho\| \sigma) \leq \frac{(1-\theta)}{\alpha'}D_{\alpha}(\rho\|\sigma) + \frac{\theta }{\gamma'} D_{\gamma}(\rho\| \sigma).$$
Equivalently, the function $1/\alpha\mapsto D_{\alpha}(\rho\| \sigma)/\alpha'$ is convex.
\end{theorem}

We now prove the duality property of quantum conditional R\'enyi entropy.

\begin{theorem}\label{thm:duality}\emph{(Duality)} For all $1/2\leq  \alpha, \beta \leq \infty$, $\alpha,\beta\neq 1$, such that $\frac{1}{\alpha} + \frac{1}{\beta}=2$, and all tripartite \emph{pure} states $\rho_{ABC}$ we have $H_{\alpha}(A|B) = -H_{\beta}(A|C)$.
\end{theorem}

\begin{proof}
By assumptions one of $\alpha, \beta$ is greater than $1$ and the other is less than $1$. So let us assume that $\alpha>1$ and $\beta<1$. Moreover, $\frac{1}{\alpha} + \frac{1}{\beta}=2$ implies that $\beta'=-\alpha'<0$. Taking these into account and using the monotonicity of the logarithm function, $H_{\alpha}(A|B) = -H_{\beta}(A|C)$ is equivalent to
\begin{align}\label{eq:inf-sup-1}
\inf_{\sigma_{B}} \| \mathcal I_A\otimes  \Gamma_{\sigma_B}^{-\frac{1}{\alpha'}}(\rho_{AB})  \|_\alpha = \sup_{\tau_C}  \| \mathcal I_A\otimes \Gamma_{\tau_C}^{-\frac{1}{\beta'}}(\rho_{AC})  \|_\beta,
\end{align}
where $\mathcal I$ denotes the identity super-operator.

Observe that  
$$\mathcal I_A\otimes  \Gamma_{\sigma_B}^{-\frac{1}{\alpha'}}(\rho_{AB})  = \mathcal I_A\otimes  \Gamma_{\sigma_B}^{-\frac{1}{\alpha'}} \otimes \tr_C (\rho_{ABC}) = \tr_C\left(   \mathcal I_A\otimes  \Gamma_{\sigma_B}^{-\frac{1}{\alpha'}}\otimes \mathcal I_C(\rho_{ABC})    \right).$$
On the other hand $\rho_{ABC}$ and then $ \mathcal I_A\otimes  \Gamma_{\sigma_B}^{-\frac{1}{\alpha'}}\otimes \mathcal I_C(\rho_{ABC})  $ are pure (but not necessarily normalized). As a result the set of eigenvalues (and singular values) of 
$$ \tr_C\left(   \mathcal I_A\otimes  \Gamma_{\sigma_B}^{-\frac{1}{\alpha'}}\otimes \mathcal I_C(\rho_{ABC})\right)  \quad \quad \text{ and }  \quad \quad \tr_{AB}\left(   \mathcal I_A\otimes  \Gamma_{\sigma_B}^{-\frac{1}{\alpha'}}\otimes \mathcal I_C(\rho_{ABC}) \right),$$
are equal, which implies that their $\alpha$-norm coincide. Repeating the same argument with the right hand side of~\eqref{eq:inf-sup-1} we find that it suffices to prove 
\begin{align}\label{eq:inf-sup-2}
\inf_{\sigma_B} \| \left(\tr_B\circ \Gamma_{\sigma_B}^{-\frac 1 {\alpha'}}\right) \otimes \mathcal I_C (\rho_{BC})   \|_{\alpha} = \sup_{\tau_C} \|  \mathcal I_B\otimes \left(  \tr_C\circ \Gamma_{\tau_C}^{- \frac 1{\beta'}}  \right) (\rho_{BC})    \|_{\beta}.
\end{align}

Now using H\"older's duality~\eqref{eq:holder-dual} the left hand side is equal to 
\begin{align*}
\inf_{\sigma_B} \| \left(\tr_B\circ \Gamma_{\sigma_B}^{-\frac 1 {\alpha'}}\right) \otimes \mathcal I_C (\rho_{BC})   \|_{\alpha} &
 = \inf_{\sigma_B} \, \sup_{\tau_C}\,  \tr\left[ \tau_C^{\frac{1}{\alpha'}} \left(\tr_B\circ \Gamma_{\sigma_B}^{-\frac 1 {\alpha'}}\right) \otimes \mathcal I_C (\rho_{BC})   \right] \\
 & = \inf_{\sigma_B} \, \sup_{\tau_C}\,  \tr\left[ I_B\otimes \tau_C^{\frac{1}{\alpha'}} \left( \Gamma_{\sigma_B}^{-\frac 1 {\alpha'}}\otimes \mathcal I_C (\rho_{BC}) \right)   \right] \\
 & = \inf_{\sigma_B} \, \sup_{\tau_C}\,  \tr\left[\rho_{BC} \left(\sigma_B^{-\frac 1{\alpha'}}\otimes \tau_C^{\frac{1}{\alpha'}}   \right) \right].\\
\end{align*}
Here for the first equation we use the fact that in~\eqref{eq:holder-dual} when $X$ is positive semi-definite we may restrict the optimization over positive semi-definite $Y$. Moreover any such $Y$ with $\|Y\|_{\alpha'}=1$ is of the form $Y=\tau^{1/\alpha'}$ where $\tau$ is a density matrix.

Again by repeating the same argument for the right hand side of~\eqref{eq:inf-sup-2} and using~\eqref{eq:inverse-holder-duality} we arrive at
$$\sup_{\tau_C} \|  \mathcal I_B\otimes \left(  \tr_C\circ \Gamma_{\tau_C}^{- \frac 1{\beta'}}  \right) (\rho_{BC})    \|_{\beta}
 = \sup_{\tau_C}\,  \inf_{\sigma_B}\, \tr\left[  \rho_{BC} \left( \sigma_B^{-\frac 1 {\alpha'}}\otimes \tau_C^{\frac 1 {\alpha'}} \right)  \right].$$
Note that here we use $\beta'= - \alpha'$. As a result it suffices to show
\begin{align}\label{eq:inf-sup-3}
\inf_{\sigma_B} \, \sup_{\tau_C}\, f(\sigma_B, \tau_C)  =  \sup_{\tau_C}\,  \inf_{\sigma_B}\, f(\sigma_B, \tau_C),
\end{align}
where 
$$f(\sigma_B, \tau_C)= \tr\left[\rho_{BC} \left(\sigma_B^{-\frac 1{\alpha'}}\otimes \tau_C^{\frac{1}{\alpha'}}   \right) \right].$$

This equation holds due to Sion's minimax theorem \cite{Sion}. The ranges of $\sigma_B, \tau_C$ are compact and convex. Moreover, $\sigma_B\mapsto  f(\sigma_B, \tau_C)$ is convex because $-1\leq -1/{\alpha'}\leq 0$ and $\sigma\mapsto \sigma^{-\frac{1}{\alpha'}}$ is operator convex \cite{Bhatia-M} . Finally, $\tau_C\mapsto f(\sigma_B, \tau_C) $ is concave because $0\leq 1/\alpha'\leq 1$ and $\tau\mapsto \tau^{1/\alpha'}$ is operator monotone and then operator concave \cite{Bhatia-M}.
\end{proof}

%***************

\section{$\alpha$-Holevo information is super-additive}\label{sec:chi}

A mutual information type function is also defined in~\cite{Wildeetal}. For a bipartite state $\rho_{AB}$ let
\begin{align}\label{eq:MI}
I_{\alpha}(A; B) = \min_{\sigma_B} D_{\alpha}(\rho_{AB}\| \rho_A\otimes \sigma_B),
\end{align}
where $\rho_A = \tr_B(\rho_{AB})$ and the infimum is taken over all density matrices $\sigma_B$. For $\alpha=1$ it is easy to see that $I_1(A;B)=I(A;B)$ is equal to the usual mutual information. We call $I_\alpha(A;B)$ the $\alpha$-R\'enyi mutual information.

Using notations developed above we have
\begin{align*}
I_{\alpha}(A;B)
& = \alpha'\log\left(\min_{\sigma_B}  \|\Gamma_{\rho_A}^{-\frac{1}{\alpha'}}\otimes \Gamma_{\sigma_B}^{-\frac 1{\alpha'}} (\rho_{AB})  \|_{\alpha}\right).
\end{align*} 
Then using H\"older's duality and following similar steps as in the proof of Theorem~\ref{thm:duality} we obtain the following. 

\begin{theorem} 
Let $\alpha> 1$ and $1/2\leq \beta<1$ such that $\frac{1}{\alpha} + \frac{1}{\beta}=2$. Also 
let $\ket\psi_{ABC}$ be a purification of $\rho_{AB}$. Then we have
\begin{align}\label{eq:max-MI}
I_\alpha(A;B) = \alpha' \log\left(   \max_{\tau_C} \norm{ \left(\text{\rm{tr}}_A \circ \Gamma_{\rho_A}^{-\frac{1}{\alpha'}}\right)\otimes \mathcal I_B\otimes (\text{\rm{tr}}_C\circ \Gamma_{\tau_C}^{\frac{1}{\alpha'}})(\ket \psi\bra \psi_{ABC}) }_{\beta}   \right).
\end{align}

\end{theorem}

\begin{theorem}\label{thm:MI-add}
$\alpha$-R\'enyi mutual information is additive for $\alpha\geq 1$, i.e., for $\rho_{AA'BB'} = \rho_{AB}\otimes \rho'_{A'B'}$ we have
$$I_{\alpha}(AA'; BB')=I_{\alpha}(A;B) + I_{\alpha}(A'; B').$$
\end{theorem}

\begin{proof}
For one direction restrict the minimization in~\eqref{eq:MI} to $\sigma_{BB'} = \sigma_{B}\otimes \sigma'_{B'}$.  For the other direction take a product purification $\ket\psi_{ABCA'B'C'}=\ket\psi_{ABC}\otimes \ket{\psi'}_{A'B'C'}$ of $\rho_{AB}\otimes \rho'_{A'B'}$ and restrict the maximization in~\eqref{eq:max-MI} to $\tau_{CC'}=\tau_C\otimes \tau_{C'}$.

\end{proof}

Using this theorem we can now answer an open question raised in~\cite{Wildeetal}. 
For a noisy quantum channel $\cN_{A\rightarrow B}$ define its $\alpha$-Holevo information by
\begin{align}\label{eq:chi-alpha}
\chi_{\alpha}(\cN) = \sup_{\rho_{XA}} I_{\alpha} (X; B),
\end{align}
where the supremum is taken over all classical-quantum (c-q) states $\rho_{XA}$ and $\rho_{XB}=\mathcal I_X\otimes\cN(\rho_{XA})$. Again for $\alpha=1$ this quantity reduces to the Holevo information ($\chi_1(\cN)=\chi(\cN)$).

\begin{theorem}\label{thm:chi-super-additive}
$\alpha$-Holevo information is super-additive for $\alpha\geq 1$, i.e., for two quantum channels $\cN_{A\rightarrow B}$ and $\cN'_{A'\rightarrow B'}$ we have
$$\chi_{\alpha}(\cN\otimes \cN')\geq \chi_{\alpha}(\cN)+ \chi_{\alpha}(\cN').$$
\end{theorem}

\begin{proof} In the definition of $\chi_{\alpha}(\cN\otimes \cN')$ restrict the supremum to tensor product states $\rho_{XA}\otimes \rho'_{X'A'}$, and use the additivity of $\alpha$-R\'enyi mutual information.

\end{proof}

\noindent\textbf{Acknowledgements.} 
The author is thankful to Payam Delgosha and the unknown referee whose comments improved the readability of the paper.
This research was in part supported by National Elites Foundation and by a grant
from IPM (No. 91810409).

%%%%%%%%%%%%%%%%%%%%%%%%%%%%%%%%%%%%%%%%%%%%%%%

\end{document}